\documentclass[12pt]{article}
\usepackage{amsmath, amssymb, amsthm, enumerate, geometry, hyperref}
\hypersetup{colorlinks=true, linkcolor=blue, citecolor=blue, urlcolor=blue}

\usepackage{enumitem}
\usepackage{xcolor}
\usepackage{booktabs}
\usepackage{caption}
\usepackage{tabularx} 

\newtheorem{theorem}{Theorem}[section]
\newtheorem{lemma}[theorem]{Lemma}
\newtheorem{corollary}[theorem]{Corollary}
\newtheorem{remark}{Remark}
\newtheorem{definition}{Definition}
\newtheorem{proposition}{Proposition}

\newcommand{\Ocal}{\mathcal{O}}

\newcommand{\de}{\,\mathrm{d}} 

\usepackage{tcolorbox}
\tcbuselibrary{skins}
\newtcolorbox{block}[1]{colback=blue!5!white,colframe=blue!75!black,
fonttitle=\bfseries,title=#1}

\title{Poisson Energy Formulation for Floorplanning: Variational Analysis and Mathematical Foundations}

\author{Wenxing Zhu, Hao Ai \\
Center for Discrete Mathematics and Theoretical Computer Science,\\
Fuzhou University, Fuzhou 350116, China\\
\texttt{wxzhu@fzu.edu.cn}}

\date{}

\begin{document}

\maketitle

\begin{abstract}
Arranging many modules within a bounded domain without overlap, central to the Electronic Design Automation (EDA) of very large-scale integrated (VLSI) circuits, represents a broad class of discrete geometric optimization problems with physical constraints. This paper develops a variational and spectral framework for Poisson energy–based floorplanning and placement in physical design. We show that the Poisson energy, defined via a Neumann Poisson equation,
is exactly the squared $H^{-1}$ Sobolev norm of the density residual, providing a functional-analytic interpretation of the classical electrostatic analogy. Through spectral analysis, we demonstrate that the energy acts as an intrinsic
low-pass filter, suppressing high-frequency fluctuations while enforcing large-scale uniformity. Under a mild low-frequency dominance assumption, we establish a quantitative
linear lower bound relating the Poisson energy to the geometric overlap area, thereby justifying its use as a smooth surrogate for the hard non-overlap constraint.
We further show that projected gradient descent converges globally to stationary points and exhibits local linear convergence near regular minima. Finally, we interpret the continuous-time dynamics as a Wasserstein-2 gradient flow,
revealing the intrinsic nonlocality and global balancing behavior of the model. These results provide a mathematically principled foundation for PDE-regularized optimization in large-scale floorplanning and related geometric layout problems.
\end{abstract}



\section{Introduction}

In the Electronic Design Automation (EDA) of Very Large Scale Integrated (VLSI) circuits, fixed-outline floorplanning and placement are crucial steps in physical design and can be viewed as canonical geometric optimization problems
involving strong nonconvexity and combinatorial complexity
\cite{Adya2003FixedOutline,Kahng2011VLSI}. Basically, both tasks require arranging a set of modules or cells within a prescribed rectangular domain such that they do not overlap, while minimizing a global objective, typically the total interconnect wire length. From a computational complexity perspective, these problems are strongly NP-hard, as they generalize the classical two-dimensional bin-packing problem \cite{GareyJohnson}. Furthermore, they are inherently large-scale, ranging from thousands of modules in floorplanning to millions of cells in placement, making them a central and enduring challenge in modern chip design. The mathematical structure of such problems, involving geometric constraints and large-scale nonconvexity, motivates their analysis from a variational perspective.

A common analytical paradigm for handling the non-overlap constraints is the continuous relaxation framework, which replaces discrete feasibility conditions by smooth, density-based penalty functionals 
\cite{4544855,Markov2015Placement,7018002}. This relaxation transforms the combinatorial problem into a differentiable
optimization problem, enabling the use of gradient-based methods. A particularly successful line of work derives such density penalties from electrostatic analogies: each module is modeled as a positive charge, the placement density represents a charge distribution, and overlaps generate a potential energy analogous to electrostatic repulsion
\cite{Lu2014DAC,Lu2015TCAD}. Formally, the potential $\phi$ satisfies a Poisson equation with Neumann boundary conditions:
\[
\left\{
\begin{aligned}
 -\Delta \phi(x) &= \rho(x), \quad x\in R,\\
 \nabla \phi \cdot \mathbf{n} &= 0, \quad x\in \partial R,\\
 \int_R \phi(x)\,\mathrm{d}x &= \int_R \rho(x)\,\mathrm{d}x = 0,
\end{aligned}
\right.
\]
where $\rho(x)$ denotes the spatial density distribution.
This PDE-based relaxation paradigm has become dominant in modern VLSI placement research
\cite{agnesina2023autodmp,gu2020dreamplace,Li2025DAC,liao2023dreamplace,
Lin2019DAC,lin2020dreamplace,xplace_tcad,lu2016eplace3d,li2021elfplace,Zhu2018ICCAD},
enabling efficient enforcement of density constraints through the electrostatic analogy. While the computational effectiveness of these methods has been well demonstrated, their mathematical foundations and variational interpretation
remain unexplored.

The Poisson equation–based formulation was recently extended to large-scale fixed-outline floorplanning by the Poisson energy framework (PeF) \cite{Li2023PeF},
which defines an explicit spatial density $\rho(x)$ and the corresponding electrostatic self-energy as
\(
\sum_i \int_{v_i} \phi(u)\,\mathrm{d}u.
\)
This serves as a smooth surrogate for the geometric non-overlap constraint. Subsequent extensions incorporate module orientation and aspect ratio \cite{Huang2023ICCAD} and address mixed-size global placement \cite{Peng2024TCAD}.
Beyond this specific context, the problem of arranging geometric objects without overlap within a bounded domain constitutes a broad class of nonconvex geometric optimization problems. The mathematical formulation and analysis developed in this paper may thus offer insights applicable to other geometric optimization settings of similar structure.

Despite the remarkable empirical success of Poisson-based methods, their mathematical foundations remain largely unexplored. This work is motivated by several fundamental questions: (i) What is the precise mathematical correspondence between the Poisson energy and the geometric non-overlap constraint?  (ii) Under what regularity or scaling conditions does the smoothed functional converges? (iii) What analytical properties of the underlying elliptic operator explain the superior global balancing behavior observed in practice?

In this paper, we develop a systematic mathematical framework for the Poisson energy formulation of fixed-outline floorplanning and placement. Our formulation departs from existing Poisson-based models
\cite{Huang2023ICCAD,Li2023PeF,Peng2024TCAD}
both in analytical structure and in variational interpretation.
By combining tools from spectral analysis, the calculus of variations,
and optimal transport, we provide formal answers to the questions above
and establish a provable theoretical foundation for PDE-regularized optimization.
The main contributions are summarized as follows:

\begin{itemize}[leftmargin=*]
    \item \textbf{Spectral characterization and low-frequency filtering.}
    We prove that the Poisson energy is the spectrally weighted
    $H^{-1}$ seminorm of the density residual.
    This reveals that the inverse Laplacian acts as an intrinsic low-pass filter
    on the density variance ($L^2$ norm),
    explaining the model’s ability to resolve large-scale (low-frequency)
    imbalances in density.

    \item \textbf{Quantitative overlap guarantees.}
    Under a mild low-frequency dominance assumption,
    we derive a uniform linear lower bound for the Poisson energy
    in terms of the geometric overlap area.
    This result provides a rigorous justification for using the Poisson energy
    as a smooth surrogate for the hard non-overlap constraint.

    \item \textbf{Convergence and stability of optimization dynamics.}
    We show that projected gradient descent on the nonconvex PeF objective
    converges globally to stationary points,
    exhibits local linear convergence near regular minima,
    and defines a Lipschitz-stable solution map.

    \item \textbf{Connection to optimal transport.}
    We establish that the continuous-time dynamics of the model
    constitute a Wasserstein-2 gradient flow of the Poisson energy.
    This connection provides a variational explanation
    for the method’s nonlocal behavior and its empirical global balancing effect.
\end{itemize}

The remainder of the paper is organized as follows.
Section~\ref{sec:formulation} introduces the mathematical formulation
and spectral characterization of the Poisson energy.
Section~\ref{sec:pef_convergence} establishes quantitative overlap bounds.
Section~\ref{sec:pgd_mollified} presents the optimization algorithm
and convergence analysis.
Section~\ref{sec:nonlocality} interprets the continuous dynamics
as a Wasserstein gradient flow.
Section~\ref{sec:conclusion} concludes the paper.

\section{Mathematical Foundations of Poisson Energy in PeF}
\label{sec:formulation}

Let $R \subset \mathbb{R}^2$ be a bounded Lipschitz domain.
We consider the residual density
$r \in L^2_0(R) := \{\, f \in L^2(R)\;|\;\int_R f\,\mathrm{d}x = 0 \,\}$,
and define the potential $\phi \in H^1(R)/\mathbb{R}$ as the unique weak solution of the Neumann problem
\[
\int_R \nabla \phi \cdot \nabla v\,\mathrm{d}x = \int_R r\,v\,\mathrm{d}x,
\qquad \forall v \in H^1(R)/\mathbb{R}.
\]
By the Lax--Milgram theorem, this problem is well-posed. The associated solution operator $G := (-\Delta_N)^{-1}$, known as the Green's operator, defines a compact, self-adjoint, and positive definite map $G : L^2_0(R) \to H^2(R) \cap L^2_0(R)$. Its operator norm is given by the reciprocal of the first non-zero Neumann eigenvalue, $\|G\|_{L^2 \to L^2} = 1/\lambda_1$.
Consequently, the Poisson energy admits the functional representation
\[
E(c) = \tfrac{1}{2}\langle r, G r \rangle_{L^2(R)}
      = \tfrac{1}{2}\int_R |\nabla \phi|^2\,\mathrm{d}x.
\]

Building on this variational setting, this section presents the mathematical formulation of the Poisson energy based floorplanning (PeF) method~\cite{Li2023PeF} and establishes its key theoretical properties. We define the continuous optimization problem, introduce the Poisson energy functional as a smooth surrogate for the non-overlap constraint, and provide a spectral analysis that justifies its effectiveness.

\subsection{The PeF Optimization Problem}

Basically, the fixed-outline floorplanning or placement problem seeks to arrange a set of hard rectangular modules $\{M_i\}_{i=1}^n$ inside a rectangular domain $R\subset \mathbb{R}^2$ to minimize the total interconnect wirelength while ensuring they do not overlap. This problem is inherently large-scale and nonconvex, combining geometric and combinatorial complexity.

The PeF methodology \cite{Li2023PeF} addresses this challenge by casting it as a continuous optimization problem. The objective is to minimize a composite function that combines a smooth wirelength term with a Poisson-based energy penalty:
\begin{equation}
\label{eq:objective}
\min_{c} F(c) = W(c) + \lambda E(c),
\qquad \lambda \ge 0,
\end{equation}
where $c = (c_1, \dots, c_n) \in \mathbb{R}^{2n}$ is the vector of all module center coordinates, $W(c)$ is the total wirelength, and $E(c)$ is the Poisson energy that penalizes non-uniform density. 


Note that our penalty term $\lambda E(c)$ is a simple rescaling of the original penalty term used in \cite{Li2023PeF}, as the two are related by a factor of two. We use $E(c)$ directly as it forms the basis of our subsequent theoretical analysis.

To enable gradient-based optimization, the non-differentiable half-perimeter wirelength (HPWL) is replaced by a smooth approximation. The wirelength of each net $\mathcal{N}$ is typically modeled as:
\begin{equation*}
W_{\mathcal{N}} = \Phi\Big(\max_{i\in \mathcal{N}} (c_i \cdot e_1) - \min_{i\in \mathcal{N}} (c_i \cdot e_1)\Big) + \Phi\Big(\max_{i\in \mathcal{N}} (c_i \cdot e_2) - \min_{i\in \mathcal{N}} (c_i \cdot e_2)\Big),
\end{equation*}
where $e_1$ and $e_2$ are the standard basis vectors in $\mathbb{R}^2$, and $\Phi$ is a smooth approximation of the identity function, such as the log-sum-exp \cite{Naylor2001} or weighted-average models \cite{Hsu2013TSV}. The total wirelength, $W(c) = \sum_{\mathcal{N}} W_{\mathcal{N}}$, is thus a smooth function of the coordinates $c$. This smoothing ensures differentiability of $W(c)$ and allows efficient gradient-based optimization.

\subsection{The Poisson Energy  Functional}

To enforce the non-overlap constraint in a differentiable manner, PeF \cite{Li2023PeF} introduces a penalty based on the deviation of the module density from a uniform distribution across the placement region. 
Specifically, for each module $M_i$, define its indicator density:
\[
\rho_i(x) =
\begin{cases}
1, & x\in M_i, \\
0, & \text{otherwise}.
\end{cases}
\]
The aggregate density and the average density are then given by
\begin{equation*}
\rho(x; c) = \sum_{i=1}^n \rho_i(x),
\qquad 
\bar\rho = \frac{1}{|R|}\sum_{i=1}^n A_i,    
\end{equation*}
and the density residual is defined as
\begin{equation}
\label{eq_residual}
r(x; c) = \rho(x; c) - \bar\rho. 
\end{equation}
This residual measures the local deviation from the ideal uniform density, which serves as the foundation for defining the Poisson energy.

For each module $M_i$, letting $M_i \subset R$ also denote its occupied region, PeF \cite{Li2023PeF} defines the module-level potential energy as
\[
N_i = \int_{M_i} \phi(u)\, \de u,
\]
where $\phi(u)$ is the solution to the Poisson equation with Neumann boundary conditions:
\begin{equation}
\label{eq_poisson}
\left\{
\begin{aligned}
 -\Delta \phi(x)  &= r(x; c), & x\in R, \\
 \nabla \phi \cdot \mathbf{n}  &= 0, & x\in \partial R, \\
 \int_R \phi(x)\, \de x  &= 0,
\end{aligned}
\right.
\end{equation}
and $r(x;c)$ is the residual density defined in \eqref{eq_residual}. The total density penalty in Equation \eqref{eq:objective} is then given by
\[
D(c) = \sum_{i=1}^n N_i.
\]

\begin{definition}
\label{def_energy}
The \emph{Poisson energy} of a floorplan $c$ is defined as the quadratic form
\[
E(c) = \frac{1}{2}\int_R r(x; c)\,\phi(x)\, \de x,
\]
where $\phi$ is the solution of the Poisson equation \eqref{eq_poisson}.
\end{definition}

The Poisson energy $E(c)$ and the module-level penalty $D(c)$ are related as follows:


\begin{proposition}
\label{thm:density_energy}
Given a floorplan $c$, the Poisson energy $E(c)$ equals one half of the sum of the module-level potential energies:
\[
E(c) = \frac 12 D(c) = \frac{1}{2}\sum_{i=1}^n N_i.
\]
\end{proposition}

\begin{proof}
From Definition \ref{def_energy},
\[
E(c) = \frac{1}{2}\int_R (\rho(x; c) - \bar\rho)\phi(x)\, \de x.
\]
The second term vanishes since $\bar\rho$ is constant and 
$\int_R \phi(x)\,\de x = 0$. Hence
\[
E(c) = \frac{1}{2}\int_R \rho(x; c)\phi(x)\, \de x
= \frac{1}{2}\sum_{i=1}^n \int_R \rho_i(x)\phi(x)\, \de x.
\]
As $\rho_i$ is the indicator function of $M_i$, this integral reduces to
$\tfrac{1}{2}\sum_{i=1}^n N_i$, completing the proof.
\end{proof}

Proposition~\ref{thm:density_energy} reveals a fundamental structural property of the Poisson energy: the global density energy decomposes exactly into one half of the sum of module-level potential energies. This result establishes a mathematical link between global and local quantities and serves as the foundation for the theoretical analysis of PeF \cite{Li2023PeF} in the subsequent sections.

\begin{lemma}[Equivalence of Energy Representations]
\label{lemma_equiv}
Let the density fluctuation be denoted by $f(x) := \rho(x;c) - \bar\rho$. Let the potential $\phi(x)$ be the unique, zero-mean solution to the Neumann boundary value problem \eqref{eq_poisson}. 
 Let $G(x,y)$ be the Neumann Green's function corresponding to the domain $R$. Then, the energy $E(c)$  has the equivalent integral representation:
\[
E(c) = \frac{1}{2} \int_R \int_R G(x,y) f(x) f(y) \, \de x \de y.
\]
\end{lemma}

\begin{proof}
The proof relies on representing the solution of the Poisson equation using the Neumann Green's function. 

Let $f(x) = \rho(x;c) - \bar\rho$. The Poisson equation with Neumann boundary conditions reads
\[
-\Delta \phi(x; c) = f(x), \quad x \in R,
\]
with $\nabla \phi \cdot \mathbf{n} = 0$ on $\partial R$ and $\int_R \phi = 0$. Since $-\Delta$ with Neumann boundary conditions is a linear operator, its inverse can be represented via the Neumann Green's function $G(x,y)$:
\[
\phi(x; c) = \int_R G(x,y) f(y)\, \de y.
\]

By the definition of the Poisson energy,
\[
E(c) = \frac{1}{2} \int_R \phi(x) f(x) \, \de x,
\]
substituting the integral representation of $\phi(x)$ gives
\[
E(c) = \frac{1}{2} \int_R \left[ \int_R G(x,y) f(y) \, \de y \right] f(x) \, \de x.
\]
Interchanging the order of integration (justified by Fubini's theorem) and rearranging the terms yields
\[
E(c) = \frac{1}{2} \int_R \int_R G(x,y) f(x) f(y) \, \de x \de y,
\]
which is exactly the integral representation claimed in the lemma.
\end{proof}

\subsection{Spectral Analysis of Poisson Energy and Density Variance}

To theoretically analyze the Poisson energy \eqref{def_energy}  in PeF~\cite{Li2023PeF}, we introduce the notion of density variance, which measures the deviation of the spatial density from its mean. Spectral analysis reveals that $E(c)$ acts as a smooth, low-frequency–dominated surrogate of this variance, providing a functional-analytic interpretation of its behavior.

\begin{definition}
\label{def:variance}
The variance of a floorplan $c$ is defined as
\begin{equation*}
\text{Var}(\rho) := \int_R (\rho(x;c) - \bar\rho)^2 \de x = \int_R r^2(x; c)  \de x.
\end{equation*}
\end{definition}

This quantity measures the deviation of the density from uniformity. The next theorem shows that the Poisson energy $E(c)$ acts as a smooth proxy for this variance.

\begin{theorem}[Poisson Energy as  Smooth Quadratic Proxy of Variance]
\label{thm:smoothproxy}
The Poisson energy of a floorplan configuration $c$ satisfies
\[
E(c) := \frac 12 \int_R \phi(x) r(x;c) \, dx = \frac 12 \langle r, G r \rangle_{L^2(R)},
\]
where $G = (-\Delta_N)^{-1}$ is the Neumann Green operator.
Moreover, its first variation with respect to the density $r$ satisfies $\delta E/\delta r = \phi = G r$, meaning that the gradient field used for optimization is a spatially smoothed version of $r$ obtained by convolution with the Poisson kernel.
\end{theorem}

\begin{proof}
The Poisson equation with Neumann BC \eqref{eq_poisson} can be written in operator form as
\[
\phi = G r, \quad G = (-\Delta_N)^{-1}, \quad r \in L^2_0(R),
\]
where $G$ is self-adjoint and positive definite on zero-mean functions. By Lemma \ref{lemma_equiv},
the Poisson energy is
\[
E(c) = \frac 12 \langle r, G r \rangle_{L^2(R)}.
\]
Since $G$ is strictly positive definite, $E(c)$ is a strictly convex quadratic form in $r$.

The functional derivative of $E(c)$ with respect to density position $r$ is
\[
\frac{\delta E}{\delta r} = \phi = G r,
\]
i.e., the gradient field used for layout optimization is the smoothed version of $r$, corresponding to convolution with the Poisson kernel. 
\end{proof}

To further quantify how the Poisson energy weights different spatial frequencies of density fluctuations, we analyze its spectral decomposition in the eigenbasis of the Neumann Laplacian.

\begin{theorem}[Spectral Representation of the Poisson Energy and Variance]
\label{thm:spectral_representation}

Consider the negative Laplacian operator with Neumann boundary conditions, $-\Delta_N$, on a finite rectangular domain $R$. Let $\{\lambda_k, \psi_k(x)\}_{k=0}^\infty$ be its eigenpairs, where $0=\lambda_0 < \lambda_1 \le \lambda_2 \le \dots$, and $\{\psi_k\}$ form a complete orthonormal basis of $L^2(R)$. 

For any residual density $r(x;c) \in L^2_0(R)$ (the subspace of zero-mean functions), we have the expansion
\[
r(x;c) = \sum_{k=1}^\infty \alpha_k \psi_k(x), 
\quad \alpha_k = \int_R r(x;c)\psi_k(x)\, \de x,
\]
where the summation starts from $k=1$ because $r$ is orthogonal to the constant eigenfunction $\psi_0$. Then:
\begin{enumerate}
    \item \textbf{Variance:}
    \[
    \mathrm{Var}(\rho) = \sum_{k=1}^\infty \alpha_k^2,
    \]
     which represents the total spectral energy of the residual density over all non-constant modes.
    
    \item \textbf{Poisson Energy:}
    \[
    E(c) = \frac{1}{2} \sum_{k=1}^\infty \frac{\alpha_k^2}{\lambda_k},
    \]
    where the factor $1/\lambda_k$ downweights high-frequency modes (large $\lambda_k$) and emphasizes low-frequency components.
\end{enumerate}
\end{theorem}

\begin{proof}
By definition,
\[
\mathrm{Var}(\rho) = \|r\|_{L^2(R)}^2 = \int_R r(x;c)^2 \, \de x.
\]
Expanding $r$ in the eigenbasis $\{\psi_k\}_{k=1}^\infty$ and applying Parseval's identity yields
\[
\|r\|_{L^2(R)}^2 = \sum_{k=1}^\infty \alpha_k^2,
\]
which proves the variance representation.

By Theorem \ref{thm:smoothproxy}, the Poisson energy can be written as
\[
E(c) = \tfrac{1}{2}\langle r, G r \rangle, \quad G = (-\Delta_N)^{-1}.
\]
On the zero-mean subspace $L^2_0(R)$, the inverse Neumann Laplacian
$G := (-\Delta_N)^{-1}$ is a self-adjoint, positive definite, and bounded operator,
with operator norm $\|G\|_{L^2\to L^2} = 1/\lambda_1$.
Hence $Gr\in L^2_0(R)$ for every $r\in L^2_0(R)$,
and the spectral series
\[
Gr=\sum_{k=1}^\infty \frac{\alpha_k}{\lambda_k}\psi_k
\]
converges in the $L^2$ sense. Consequently, the inner products and termwise summations
used below are mathematically justified.

For each eigenfunction $\psi_k$ with $k \ge 1$, we have 
\[
G \psi_k = \tfrac{1}{\lambda_k}\psi_k.
\]
Applying $G$ to the expansion of $r$ gives
\[
Gr = \sum_{k=1}^\infty \tfrac{\alpha_k}{\lambda_k}\psi_k.
\]
Thus
\[
E(c) = \tfrac{1}{2} \Big\langle \sum_{i=1}^\infty \alpha_i \psi_i,\,
\sum_{j=1}^\infty \tfrac{\alpha_j}{\lambda_j}\psi_j \Big\rangle
= \tfrac{1}{2} \sum_{k=1}^\infty \frac{\alpha_k^2}{\lambda_k},
\]
where orthonormality of $\{\psi_k\}$ eliminates cross terms. This proves the second part. 

Moreover, by integration by parts (using the Neumann boundary condition $\partial_n\phi=0$),
we obtain the equivalent expression
\begin{equation}
\label{eq:ecnonnegative}
E(c)=\tfrac12\int_R r\phi\,\de x
   = \tfrac12\int_R |\nabla\phi|^2\,\de x \ge 0,    
\end{equation}
which also directly shows the non-negativity of the Poisson energy.
\end{proof}

Equation~\eqref{eq:ecnonnegative} also shows that $E(c)$ defines the squared $H^{-1}$ norm of $r$, i.e., $E(c)=\tfrac{1}{2}\|r\|_{H^{-1}}^2$. Moreover, the $1/\lambda_k$ weighting implies that $E(c)$ penalizes long-wavelength (low-frequency) deviations more strongly, thus enforcing large-scale uniformity while smoothing high-frequency noise.

\subsection{Corollaries and Discussions}

The spectral decomposition in Theorem~\ref{thm:spectral_representation} not only clarifies why the Poisson energy emphasizes low-frequency density deviations but also yields direct analytical insights into the global structure of the density field.
In particular, it shows that $E(c)$ is non-negative and vanishes exactly when all residual density fluctuations are zero, i.e., when the placement is perfectly uniform.

\begin{corollary}
\label{cor:zero_equivalence}
Under the assumptions of Theorem~\ref{thm:spectral_representation} (bounded rectangular domain $R$, Neumann Laplacian $-\Delta_N$, and $r(\cdot;c)\in L^2_0(R)$), $E(c)\ge 0$, and we have the equivalence:
\[
\mathrm{Var}(\rho)=0 \text{ if and only if } E(c)=0.
\]
\end{corollary}

\begin{proof}
Recall the spectral representations from Theorem~\ref{thm:spectral_representation}:
\[
\mathrm{Var}(\rho)=\sum_{k=1}^\infty \alpha_k^2,
\qquad
E(c)=\frac{1}{2}\sum_{k=1}^\infty \frac{\alpha_k^2}{\lambda_k},
\]
where $\{\lambda_k\}_{k\ge1}$ are the positive eigenvalues of $-\Delta_N$ and $\alpha_k=\langle r,\psi_k\rangle_{L^2(R)}$.
$E(c)\ge 0$ follows from Equation \eqref{eq:ecnonnegative} in the proof of Theorem \ref{thm:spectral_representation}.

Next, if $\mathrm{Var}(\rho)=0$, then
\[
0=\sum_{k=1}^\infty \alpha_k^2.
\]
Since each term $\alpha_k^2\ge0$, it follows that $\alpha_k=0$ for every $k\ge1$. Hence the series defining $E(c)$ is identically zero:
\[
E(c)=\frac{1}{2}\sum_{k=1}^\infty \frac{\alpha_k^2}{\lambda_k}=0.
\]

Conversely, if $E(c)=0$, then
\[
0=\frac{1}{2}\sum_{k=1}^\infty \frac{\alpha_k^2}{\lambda_k}.
\]
Because each $\lambda_k>0$ and each $\alpha_k^2/\lambda_k\ge0$, every term must vanish: $\alpha_k^2/\lambda_k=0$, hence $\alpha_k=0$ for all $k\ge1$. Thus
\[
\mathrm{Var}(\rho)=\sum_{k=1}^\infty \alpha_k^2 = 0.
\]

Therefore $\mathrm{Var}(\rho)=0$ if and only if $E(c)=0$, as claimed.
\end{proof}

The spectral representation also allows us to bound the Poisson energy in terms of the variance. In particular, the following corollary establishes an explicit upper bound.

\begin{corollary}[Energy--Variance Upper Bound]
\label{cor:energy_variance}
Let $\lambda_1$ be the first nonzero eigenvalue of $-\Delta_N$ on $R$. Then the Poisson energy $E(c)$ and the variance $\mathrm{Var}(\rho)$ satisfy the inequality
\[
E(c) \;\le\; \frac{1}{2\lambda_1}\,\mathrm{Var}(\rho).
\]
\end{corollary}

\begin{proof}
From Theorem~\ref{thm:spectral_representation}, we have
\[
E(c) = \frac{1}{2}\sum_{k=1}^\infty \frac{\alpha_k^2}{\lambda_k}, 
\quad
\mathrm{Var}(\rho) = \sum_{k=1}^\infty \alpha_k^2.
\]
Since $\lambda_k \ge \lambda_1$ for all $k \ge 1$, it follows that
\[
\frac{1}{\lambda_k} \le \frac{1}{\lambda_1}, \quad \forall k \ge 1.
\]
Therefore,
\[
E(c) = \frac{1}{2}\sum_{k=1}^\infty \frac{\alpha_k^2}{\lambda_k}
\le \frac{1}{2\lambda_1}\sum_{k=1}^\infty \alpha_k^2
= \frac{1}{2\lambda_1}\,\mathrm{Var}(\rho).
\]
This proves the claim.
\end{proof}

Corollary~\ref{cor:energy_variance} shows that the Poisson energy is always upper-bounded by a scaled variance, with the scaling determined by the smallest nonzero eigenvalue. 
Conversely, Corollary~\ref{cor:lower_bound} introduces a complementary lower bound based on the first $N$ spectral modes. These results characterize how $E(c)$ is spectrally bounded by the variance, rather than tightly constraining it in all cases.

\begin{corollary}[Mode-Truncated Lower Bound]
\label{cor:lower_bound}
Let the notation and assumptions be as in Theorem~\ref{thm:spectral_representation}, and let $\lambda_1\le\lambda_2\le\cdots$ be the positive eigenvalues of $-\Delta_N$. 
Then for any integer $N\ge1$ the Poisson energy $E(c)$ satisfies the lower bound
\[
E(c) \;\ge\; \frac{1}{2\lambda_N}\sum_{k=1}^N \alpha_k^2,
\]
where $\alpha_k=\langle r,\psi_k\rangle_{L^2(R)}$ are the spectral coefficients of the residual density $r$.
\end{corollary}

\begin{proof}
From Theorem~\ref{thm:spectral_representation} we have the spectral formula
\[
E(c)=\frac{1}{2}\sum_{k=1}^\infty \frac{\alpha_k^2}{\lambda_k}.
\]
For fixed $N\ge1$ and each $k\in\{1,\dots,N\}$ we have $\lambda_k\le\lambda_N$, hence $\frac{1}{\lambda_k}\ge\frac{1}{\lambda_N}$. Restricting the infinite sum to the first $N$ terms gives
\[
E(c) \;=\; \frac{1}{2}\sum_{k=1}^\infty \frac{\alpha_k^2}{\lambda_k}
\;\ge\; \frac{1}{2}\sum_{k=1}^N \frac{\alpha_k^2}{\lambda_k}
\;\ge\; \frac{1}{2}\sum_{k=1}^N \frac{\alpha_k^2}{\lambda_N}
\;=\; \frac{1}{2\lambda_N}\sum_{k=1}^N \alpha_k^2,
\]
which is the claimed inequality.
\end{proof}

\begin{remark}[Interpretation and Significance]
The spectral representation in Theorem~\ref{thm:spectral_representation} highlights the structural distinction between the variance and the Poisson energy. While variance ($\sum \alpha_k^2$) weights all spectral modes equally, the Poisson energy ($ \frac{1}{2}\sum \alpha_k^2/\lambda_k$) is a spectrally weighted metric. 
The factor $1/\lambda_k$ acts as an intrinsic \textbf{low-pass weighting}, emphasizing low-frequency modes (corresponding to large-scale imbalances) while attenuating high-frequency components (localized fluctuations).
This property makes $E(c)$ an ideal objective for global placement. It naturally guides the optimization to first resolve macroscopic congestion before refining local details. 

Furthermore, the corollaries provide quantitative characterization for using $E(c)$ as a proxy for variance. The Energy-Variance Inequality (Corollary~\ref{cor:energy_variance}), in particular, formalizes this relationship, showing that the energy is controlled by the variance, with the strength of this connection determined by the domain's geometry via the spectral gap $\lambda_1$. The Mode-Truncated Lower Bound (Corollary~\ref{cor:lower_bound}), offers a practical tool for estimation and diagnosis. It allows for a quick lower estimate of the total energy by considering only the first $N$ principal modes, which can help identify whether density nonuniformities arise primarily from large-scale imbalances or localized fluctuations. 
\end{remark}

Together, these results provide a coherent theoretical foundation that formally justifies the use of Poisson energy as a smooth, analytically tractable surrogate for overlap avoidance in floorplanning.

\section{Mollified Module and Poisson Energy Bound }
\label{sec:pef_convergence}

In the previous section, we analyzed the Poisson energy $E(c)$ as a smooth quadratic proxy for the variance of module density. This provided insight into how the energy penalizes global density imbalances and guides gradient-based optimization. However, the analysis does not directly quantify geometric module overlaps, which is the ultimate physical constraint we aim to eliminate. 

This section establishes a rigorous connection between the mollified Poisson energy and the geometric overlap measure.
We first define mollified densities and the corresponding Poisson energy. Then we establish fundamental properties showing that the energy detects overlap qualitatively. Finally, we derive a quantitative, linear lower bound connecting $E_\varepsilon(c)$ to the total overlap area under a physically motivated low-frequency assumption. These results provide the theoretical guarantees that justify using the mollified Poisson energy as a robust surrogate for overlap avoidance in floorplanning.

\subsection{Assumptions and Preliminaries}
\label{subsec:assumptions}

We now formalize the setup for analyzing the mollified Poisson energy $E_\varepsilon(c)$. Consider a floorplan region $R \subset \mathbb{R}^2$ and a set of modules $\{M_i\}_{i=1}^n$, which are all assumed to be bounded Lipschitz domains with well-defined areas $A_i$ and perimeters $P(M_i)$. The total perimeter is the sum of the perimeters of all modules: $P_\Sigma := \sum_{i=1}^n P(M_i)$.

Our main theoretical results rely on the following assumptions:

\begin{enumerate}[label=(A\arabic*)]

    \item \textbf{Density Feasibility:} The average density satisfies $\bar\rho < 2$. This is a mild condition (typically, $\bar\rho \le 1$ in any placeable design) that ensures certain constants in our bounds are well-defined and positive. \label{assum_a1}

    \item \textbf{Existence of an Optimum:} There exists an optimal, non-overlapping floorpan, denoted $c^{\mathrm{opt}}$, which has the minimum possible wirelength. \label{assum_a2}
    
    \item \textbf{Mollifier Scale:} The mollifier radius $\varepsilon > 0$ is strictly smaller than the minimal inradius of all modules, i.e., $\varepsilon < \min_i \mathrm{inrad}(M_i)$. This technical condition ensures that the ``eroded" module shapes used in the proofs are well-defined. \label{assum_a3}

 \item \textbf{Compact Configuration Space:} The space of allowed floorplan $\mathcal{C} = \{c = (c_1, \dots,$ $c_n)\}$ is a compact subset of $\mathbb{R}^{2n}$. This is a standard physical assumption, ensuring that modules are constrained to a bounded region and cannot move infinitely far away. This assumption is critical for guaranteeing the existence of certain minima in the proof.  \label{assum_a4}  
\end{enumerate}

Next, we introduce key definitions for mollification:

\begin{enumerate}[label=(D\arabic*)]
\item \textbf{Mollifier:} Fix a mollifier $\eta \in C_c^\infty(\mathbb R^2)$ with
\[
\eta \ge 0, \quad \int_{\mathbb R^2} \eta(x)\,\de x = 1, \quad \operatorname{supp} \eta \subset B(0,1).
\]

For $\varepsilon>0$, define the scaled mollifier
\[
\eta_\varepsilon(x) := \frac{1}{\varepsilon^2}\, \eta\Big(\frac{x}{\varepsilon}\Big),
\]
which satisfies $\operatorname{supp} \eta_\varepsilon \subset B(0,\varepsilon)$ and $\int \eta_\varepsilon = 1$.

\item \textbf{Mollified Indicator Function:}   The mollified indicator function of a module $M_i \subset \mathbb R^2$ is defined by
\[
\psi_i^\varepsilon(x) := (\eta_\varepsilon * \mathbf{1}_{M_i})(x)
= \int_{\mathbb R^2} \eta_\varepsilon(x-y) \, \mathbf{1}_{M_i}(y)\, \de y.
\]
Equivalently, using the change of variables $z = \frac{x-y}{\varepsilon}$, we can write
\[
\psi_i^\varepsilon(x) = \int_{B(0,1)} \eta(z) \, \mathbf{1}_{M_i}(x - \varepsilon z)\, \de z.
\]

This shows that $\psi_i^\varepsilon$ is a smooth approximation of the characteristic function $\mathbf{1}_{M_i}$, with
\[
0 \le \psi_i^\varepsilon(x) \le 1, \quad 
\psi_i^\varepsilon(x) \to \mathbf{1}_{M_i}(x) \text{ as } \varepsilon \to 0.
\]

\item \textbf{Mollified Density:} For a floorplan $c = (c_1,\dots,c_n) \in \mathbb R^{2n}$, define the \emph{mollified density} as
\[
\rho_\varepsilon(x;c) := \sum_{i=1}^n \psi_i^\varepsilon(x-c_i), 
\qquad x \in R,
\]
where $\psi_i^\varepsilon$ is the mollified indicator of the $i$-th module.

\item \textbf{Mollified Poisson Energy:} 
The \emph{Poisson energy} of the mollified  floorplan $c$ is then defined by
\begin{equation}
\label{eq_varenergy}
E_\varepsilon(c) := \frac{1}{2} \int_R |\nabla \phi_\varepsilon(x;c)|^2 \, \de x
= \frac{1}{2} \int_R \phi_\varepsilon(x;c)\, \big(\rho_\varepsilon(x;c) - \bar\rho\big)\, \de x,    
\end{equation}
where $\phi_\varepsilon$ is the unique solution to the Poisson problem with Neumann boundary conditions:
\begin{equation}
\label{eq_varpoisson}
\begin{cases}
-\Delta \phi_\varepsilon(x;c) = \rho_\varepsilon(x;c) - \bar\rho, & \text{in } R, \\
\frac{\partial \phi_\varepsilon}{\partial n}(x;c) = 0, & \text{on } \partial R, \\
\int_R \phi_\varepsilon(x;c) \, \de x = 0.
\end{cases}
\end{equation}

Here, $\bar\rho = \frac{1}{|R|}\sum_{i=1}^n A_i$ is the average density.  
This smooth energy measures deviation from uniform density and depends smoothly on $c$ due to the mollification.
\end{enumerate}

This mollified formulation provides a smooth surrogate for the otherwise discontinuous overlap constraint. The mollified density $\rho_\varepsilon$ and the corresponding Poisson energy $E_\varepsilon(c)$ retain the essential geometric information of module placement while being differentiable with respect to $c$. This makes $E_\varepsilon(c)$ suitable for gradient-based optimization.

\subsection{Relating Mollified Poisson Energy to Module Overlap}
\label{subsec:energy_overlap}

To quantify how the mollified Poisson energy $E_\varepsilon(c)$ reflects geometric overlap, we first introduce the notions of \emph{eroded module shapes} and \emph{eroded overlap}. These constructions allow us to connect the smooth, gradient-friendly energy to the underlying non-smooth geometric intersection of modules.

For each module $M_i$, define its \emph{eroded shape} at scale $\varepsilon>0$ by
\[
M_i^{-\varepsilon} := \{ x \in M_i : \operatorname{dist}(x, \mathbb{R}^2 \setminus M_i) \ge \varepsilon \},
\]
i.e., the set of points in $M_i$ whose distance to the complement is at least $\varepsilon$.

For a floorplan $c = (c_1,\dots,c_n)$, define the \emph{eroded overlap set} by
\[
\mathcal O_\varepsilon(c) := \bigcup_{i<j} \big( (M_i^{-\varepsilon} + c_i) \cap (M_j^{-\varepsilon} + c_j) \big).
\]

The \emph{geometric overlap set} is
\[
\mathcal O(c) := \bigcup_{i<j} \big( (M_i + c_i) \cap (M_j + c_j) \big),
\]
whose area is denoted $|\mathcal O(c)|$.

The following lemma establishes a basic but crucial link between the eroded and geometric overlaps, providing a foundation for connecting overlap to Poisson energy.

\begin{lemma}[Bound on Eroded Overlap Area]
\label{lem:areabound}
Let $\{M_i\}_{i=1}^n$ be a collection of bounded domains in $\mathbb{R}^2$ with finite perimeter $P(M_i)$ in the sense of De Giorgi (i.e., sets of finite perimeter). Let the total perimeter be $P_\Sigma := \sum_{i=1}^n P(M_i)$. For a placement vector $c=(c_1, \dots, c_n)$, define the geometric overlap set $\mathcal{O}(c) := \bigcup_{i<j} (M_i+c_i) \cap (M_j+c_j)$ and the eroded overlap set $\mathcal{O}_\varepsilon(c) := \bigcup_{i<j} (M_i^{-\varepsilon}+c_i) \cap (M_j^{-\varepsilon}+c_j)$, where $M_i^{-\varepsilon} := M_i \ominus B_\varepsilon$ is the Minkowski erosion by a disk of radius $\varepsilon > 0$.

Then, the areas of these sets satisfy the following inequality:
\[
|\mathcal{O}_\varepsilon(c)| \ge \big(|\mathcal{O}(c)| -  \varepsilon P_\Sigma\big)_+,
\]
where $(t)_+ := \max\{t, 0\}$.
\end{lemma}

The proof, which relies on standard inequalities from geometric measure theory, is provided in Appendix \ref{app:proof_of_lemma3.1}.

\begin{lemma}[Zero Mean of the Density Fluctuation]
\label{lem:zeromean}
Let $\{M_i\}_{i=1}^n$ be a collection of modules in $\mathbb{R}^2$ with corresponding areas $\{A_i\}_{i=1}^n$, placed in a larger rectangle $R$ with area $|R|$. Let the mollified density be defined as $\rho_\varepsilon(x) := \sum_{i=1}^n \psi_i^\varepsilon(x-c_i)$, where $\psi_i^\varepsilon$ is the mollified indicator function of $M_i$. Moreover, let the average density be the constant $\bar\rho := \frac{1}{|R|} \sum_{i=1}^n A_i$.

The density fluctuation function, defined as $f_\varepsilon(x) := \rho_\varepsilon(x) - \bar\rho$, has a mean value of zero over the domain $R$. That is:
\[
\int_R f_\varepsilon(x) \, \de x = 0.
\]
\end{lemma}

\begin{proof}
We prove the statement by direct integration of the function $f_\varepsilon$ over the domain $R$. By the definition of $f_\varepsilon$, we can split the integral into two parts:
\[
\int_R f_\varepsilon(x) \, \de x = \int_R \rho_\varepsilon(x) \, \de x - \int_R \bar\rho \, \de x.
\]
We evaluate each term on the right-hand side separately.

First, we consider the integral of the mollified density $\rho_\varepsilon$. By its definition and the linearity of integration, we have:
\begin{align*}
    \int_R \rho_\varepsilon(x) \, \de x &= \int_R \sum_{i=1}^n \psi_i^\varepsilon(x-c_i) \, \de x \\
    &= \sum_{i=1}^n \int_R \psi_i^\varepsilon(x-c_i) \, \de x.
\end{align*}
The integral of a mollified indicator function $\psi_i^\varepsilon = \eta_\varepsilon * \mathbf{1}_{M_i}$ equals the area of the original set $M_i$. This follows from Fubini's theorem and the fact that the mollifier $\eta_\varepsilon$ has a unit integral:
\[
\int_{\mathbb{R}^2} \psi_i^\varepsilon(x) \, \de x = \left(\int_{\mathbb{R}^2} \eta_\varepsilon(y)\,\de y\right) \left(\int_{\mathbb{R}^2} \mathbf{1}_{M_i}(z)\,\de z\right) = 1 \cdot A_i = A_i.
\]
By the translation invariance of the integral, it follows that $\int_R \psi_i^\varepsilon(x-c_i) \, \de x = A_i$. Therefore, the first term is the sum of all module areas:
\[
\int_R \rho_\varepsilon(x) \, \de x = \sum_{i=1}^n A_i.
\]

Next, we consider the integral of the average density $\bar\rho$. Since $\bar\rho$ is a constant with respect to the integration variable $x$, its integral over $R$ is the constant multiplied by the area of $R$:
\[
\int_R \bar\rho \, \de x = \bar\rho \cdot |R| = \left( \frac{1}{|R|} \sum_{i=1}^n A_i \right) \cdot |R| = \sum_{i=1}^n A_i.
\]

Finally, substituting the results for both terms back into the original expression, we find:
\[
\int_R f_\varepsilon(x) \, \de x = \sum_{i=1}^n A_i - \sum_{i=1}^n A_i = 0.
\]
This completes the proof.
\end{proof}

Lemma~\ref{lem:zeromean} guarantees that the mollified density fluctuation $f_\varepsilon$ is zero-mean, enabling a direct link between its $L^2$-norm, the Poisson energy $E_\varepsilon(c)$, and the geometric overlap. Before analyzing quantitative lower bounds, we note a simpler fact: any non-zero geometric overlap produces strictly positive Poisson energy. Theorem~\ref{lem:qualitative_overlap} formalizes this, showing that $E_\varepsilon(c)$ acts as a continuous overlap detector without relying on spectral truncation or constants.


\begin{theorem}[Qualitative Overlap Detection via Poisson Energy]
\label{lem:qualitative_overlap}
Let $R\subset\mathbb{R}^2$ be a bounded domain with a smooth boundary and $\{M_i\}$ be bounded Lipschitz modules. For the geometric overlap $\mathcal O(c)$ and the Poisson energy $E_\varepsilon(c)$ as in the preliminaries, if the geometric overlap is non-empty, then the Poisson energy is strictly positive. That is:
\[
|\mathcal O(c)| > 0 \implies E_\varepsilon(c) > 0.
\]
\end{theorem}

\begin{proof}
By Assumption \ref{assum_a3}, the mollifier radius \(\varepsilon\) is smaller than the minimal inradius of all modules (this ensures that eroded modules are non-empty).  
Let the eroded overlap region be
\[
A := \mathcal{O}_\varepsilon(c) = \bigcup_{i\neq j} \big(M_i \ominus B_\varepsilon \cap M_j \ominus B_\varepsilon\big),
\]
where \(\ominus B_\varepsilon\) denotes the Minkowski erosion by a disk of radius \(\varepsilon\).  
If the geometric overlap \(|\mathcal O(c)|>0\), then for sufficiently small \(\varepsilon\), the eroded overlap \(A\) is non-empty, i.e., \(|A|>0\).

On the set \(A\), at least two mollified indicators \(\psi_i^\varepsilon\) overlap, so the mollified density satisfies
\[
\rho_\varepsilon(x) = \sum_i \psi_i^\varepsilon(x) \ge 2.
\]
Since the average density satisfies \(\bar\rho < 2\) (see Assumption \ref{assum_a1}), it follows that
\begin{equation}
\label{eq:fvar}
f_\varepsilon(x) = \rho_\varepsilon(x) - \bar\rho \ge 2 - \bar\rho > 0 \quad \text{for all } x \in A.    
\end{equation}
Hence \(f_\varepsilon \not\equiv 0\) in \(L^2(R)\), 
implying that the density fluctuation is nontrivial.


Let $\{\psi_k\}_{k\ge0}$ be the complete orthonormal set of Neumann eigenfunctions of $-\Delta$ on $R$, with corresponding eigenvalues $0=\lambda_0<\lambda_1\le \lambda_2 \le \cdots$. By Lemma \ref{lem:zeromean}, $f_\varepsilon$ has zero mean, hence it lies in the space orthogonal to the constant eigenfunction $\psi_0$. We can expand $f_\varepsilon$ in the eigenbasis:
\[
f_\varepsilon = \sum_{k\ge 1} a_k \psi_k, \quad \text{where } a_k = \langle f_\varepsilon, \psi_k \rangle.
\]
The Poisson energy is half the squared $H^{-1}$ norm of $f_\varepsilon$, which is given by the spectral representation:
\[
2 E_\varepsilon(c) = \|f_\varepsilon\|_{H^{-1}}^2 = \sum_{k\ge 1} \frac{a_k^2}{\lambda_k}.
\]

 
From Equation \eqref{eq:fvar}, we know that $f_\varepsilon$ is not the zero function. Hence by Parseval's identity, $\|f_\varepsilon\|_{L^2}^2 = \sum_{k\ge 1} a_k^2 > 0$. This implies that there must be at least one coefficient, say $a_{k_0}$ for some $k_0 \ge 1$, that is non-zero.

Consequently, the squared coefficient $a_{k_0}^2$ is strictly positive. Since $\lambda_{k_0}$ (for $k_0 \ge 1$) is also strictly positive, the term $a_{k_0}^2 / \lambda_{k_0}$ is strictly positive.
The full sum for the energy is a sum of this strictly positive term and other non-negative terms ($a_k^2/\lambda_k \ge 0$ for all $k$). Therefore, the sum itself must be strictly positive:
\[
2 E_\varepsilon(c) = \sum_{k\ge 1} \frac{a_k^2}{\lambda_k} \ge \frac{a_{k_0}^2}{\lambda_{k_0}} > 0.
\]
This implies that $E_\varepsilon(c) > 0$, which completes the proof.
\end{proof}

Theorem~\ref{lem:qualitative_overlap} establishes that the Poisson energy acts as a \emph{qualitative certificate} of overlap: it is strictly positive whenever the geometric overlap is non-empty. Building on this, the following Theorem~\ref{thm:quant_overlap} provides a \emph{quantitative} guarantee, giving a linear lower bound on $E_\varepsilon(c)$ proportional to the overlap area $|\mathcal{O}(c)|$ with a placement-independent constant $C>0$. Deriving this bound requires analyzing the mollified density variance and its relation to the area of the \emph{eroded overlap region}. These results demonstrate that the Poisson energy serves as a quantitative surrogate for geometric overlap, providing a formal justification for its use as a smooth objective in floorplan optimization.

\begin{theorem}[Overlap Lower Bound Under Low-Frequency Dominance]
\label{thm:quant_overlap}
Let $f_\varepsilon = \rho_\varepsilon(\cdot;c)-\bar\rho$ be the mollified residual density of a floorplan $c$, and let $\alpha_k = \langle f_\varepsilon, \psi_k \rangle$ be its spectral coefficients with respect to the Neumann Laplacian eigenpairs $(\lambda_k, \psi_k)$. 

Assume that for some integer $N \ge 1$ and fraction $\beta \in (0, 1]$, the \textbf{low-frequency dominance condition} holds:
\begin{equation}
\label{eq:lowfreq_cond}
\sum_{k=1}^N \alpha_k^2 \ge \beta \, \mathrm{Var}(\rho_\varepsilon).    
\end{equation}
Then, provided $\bar\rho < 2$, the Poisson energy $E_\varepsilon(c)$ satisfies the linear lower bound:
\[
E_\varepsilon(c) \ge C \cdot \big(|\mathcal O(c)| - \varepsilon P_\Sigma\big)_+,
\]
where the constant $C = \frac{\beta(2-\bar\rho)^2}{2\lambda_N}$ is positive and independent of $\varepsilon$.
\end{theorem}

\begin{proof}
The mollified Poisson energy admits the spectral representation
\[
E_\varepsilon(c) = \frac{1}{2}\sum_{k=1}^\infty \frac{\alpha_k^2}{\lambda_k} \ge \frac{1}{2}\sum_{k=1}^N \frac{\alpha_k^2}{\lambda_k}.
\]
Since $\lambda_1 \le \dots \le \lambda_N$, we have $\frac{1}{\lambda_k} \ge \frac{1}{\lambda_N}$ for all $k=1,\dots,N$, giving
\[
E_\varepsilon(c) \ge \frac{1}{2\lambda_N}\sum_{k=1}^N \alpha_k^2.
\]
Applying the low-frequency dominance assumption \eqref{eq:lowfreq_cond}, $\sum_{k=1}^N \alpha_k^2 \ge \beta \, \mathrm{Var}(\rho_\varepsilon)$, we obtain
\begin{equation}
\label{eq:energyin}
E_\varepsilon(c) \ge \frac{\beta}{2\lambda_N} \, \mathrm{Var}(\rho_\varepsilon).    
\end{equation}

Next, consider the eroded overlap region $A(c) = \mathcal{O}_\varepsilon(c)$. In this region the mollified density satisfies $\rho_\varepsilon(x) \ge 2$, so by Assumption~\ref{assum_a1}, the residual density $f_\varepsilon = \rho_\varepsilon - \bar\rho$ is bounded below:
\[
f_\varepsilon(x) \ge 2 - \bar\rho =: m > 0, \quad \forall x \in A(c).
\]
Thus, the variance satisfies
\[
\mathrm{Var}(\rho_\varepsilon) = \int_R f_\varepsilon^2(x)\, dx \ge \int_{A(c)} f_\varepsilon^2(x)\, dx \ge m^2 |A(c)|.
\]

Substituting the variance bound into the energy inequality \eqref{eq:energyin} yields
\[
E_\varepsilon(c) \ge \frac{\beta}{2\lambda_N} \, \mathrm{Var}(\rho_\varepsilon) \ge \frac{\beta m^2}{2\lambda_N} |A(c)| =: C |A(c)|,
\]
with $C = \frac{\beta(2-\bar\rho)^2}{2\lambda_N} > 0$.  
Finally, applying the eroded overlap area bound from Lemma~\ref{lem:areabound}, $|A(c)| \ge \big(|\mathcal O(c)| - \varepsilon P_\Sigma\big)_+$, we conclude
\[
E_\varepsilon(c) \ge C \, \big(|\mathcal O(c)| - \varepsilon P_\Sigma\big)_+,
\]
which completes the proof.
\end{proof}

Theorem~\ref{thm:quant_overlap} provides a characterization for using the mollified Poisson energy $E_\varepsilon(c)$ as a surrogate for geometric non-overlap. When the low-frequency components dominate the density deviation, $E_\varepsilon(c)$ gives a quantitative lower bound proportional to the actual overlap area, up to a small perimeter correction. This establishes a clear theoretical link between the continuous optimization landscape and the discrete non-overlap constraints inherent in floorplanning. However, if overlap occurs only on very fine scales (high-frequency modes), then $E_\varepsilon$ may be small despite a nontrivial geometric overlap. Hence Theorem \ref{thm:quant_overlap} applies to regimes where the overlap manifests as low-to-medium frequency density deviations (typical in early optimization stages), but not to all possible placements.

\paragraph{On the low-frequency dominance assumption.}
The key assumption of Theorem~\ref{thm:quant_overlap} is the low-frequency dominance condition \eqref{eq:lowfreq_cond}. Physically, the residual density $r(x;c)$ can be interpreted as an error signal, where:

\begin{itemize}
\item Low-frequency spectral components (small eigenvalues $\lambda_k$) correspond to \emph{large-scale, smooth density imbalances}, such as regional crowding. These are the primary targets of a global placement optimizer.

\item High-frequency components correspond to \emph{small-scale, sharp local fluctuations}, e.g., minor overlaps between neighboring modules in a nearly-optimal configuration.
\end{itemize}

The assumption asserts that a significant fraction $\beta$ of the total density variance arises from these large-scale imbalances. This is a realistic characterization of a layout in early or intermediate optimization stages. Consequently, the condition is not a restrictive limitation but rather specifies the regime where the Poisson energy is most effective.

\section{Gradient Descent for Mollified  Floorplanning}
\label{sec:pgd_mollified}

Building on the theoretical guarantees established in Section~\ref{sec:pef_convergence}, where the mollified Poisson energy $E_\varepsilon(c)$ was shown to provide both qualitative and quantitative measures of geometric overlap, we now describe an optimization framework for mollified fixed-outline floorplanning. The floorplanning task is formulated as the minimization of a composite objective that balances total wirelength $W(c)$ with the mollified Poisson energy:
\begin{equation}
\label{eq:mfof}
\min_{c \in \mathcal{C}_R} F_{\varepsilon, \lambda}(c) := W(c) + \lambda E_\varepsilon(c),
\end{equation}
where $\lambda \ge 0$ is a penalty parameter and $\mathcal{C}_R$ denotes the feasible set of module centers constrained to lie entirely within the floorplanning region $R$. Specifically, for each module $M_i$ with center $c_i=(x_i, y_i)$ and dimensions $(w_i, h_i)$, the feasible coordinates satisfy
\[
x_i \in [w_i/2, W - w_i/2], \quad y_i \in [h_i/2, H - h_i/2],
\]
where $[0, W]\times [0,H]$ defines the rectangular domain $R$.

The smoothness of $E_\varepsilon(c)$ guarantees that $F_{\varepsilon, \lambda}(c)$ is continuously differentiable, thereby enabling the use of gradient-based optimization methods. We adopt Projected Gradient Descent (PGD) to handle the boundary constraints:
\begin{equation}
\label{eq:pgd}
c^{k+1} = \Pi_{\mathcal{C}_R}\big(c^k - \eta_k \nabla F_{\varepsilon, \lambda}(c^k)\big),
\end{equation}
where $\eta_k>0$ is the step size and $\Pi_{\mathcal{C}_R}$ denotes the Euclidean projection onto $\mathcal{C}_R$. This update guarantees that all iterates remain feasible while descending the objective.

\subsection{Convergence}
\label{sec:convergence}
Having established in the previous section that the mollified Poisson energy $E_\varepsilon(c)$ provides both qualitative and quantitative guarantees against overlap, we now turn to the convergence behavior of the optimization algorithm used to minimize the composite objective $F_{\varepsilon,\lambda}(c)$. In particular, we analyze the theoretical convergence properties of the PGD scheme introduced in \eqref{eq:pgd}.

The convergence of PGD for smooth objectives over convex feasible sets is well-established in optimization theory. For the mollified floorplanning objective \eqref{eq:mfof}, where $F_{\varepsilon,\lambda}(c)$ is continuously differentiable with an $L$-Lipschitz gradient and $\mathcal{C}_R$ is closed and convex, it follows that if the step sizes $\eta_k$ satisfy the Robbins--Monro conditions ($\sum_k \eta_k = \infty$ and $\sum_k \eta_k^2 < \infty$), any limit point of $\{c^k\}$ is a stationary point.

Moreover, with a fixed step size $\eta \le 1/L$, PGD achieves a sublinear convergence rate $\mathcal{O}(1/K)$ in the squared norm of the gradient mapping, even when $F_{\varepsilon,\lambda}$ is nonconvex. See, e.g., \cite{Beck2017} for a detailed exposition of these results.

\vspace{0.5em}
Beyond global asymptotic guarantees, one can further establish a local \emph{linear convergence rate} around a strict local minimizer, as formalized below.
\begin{theorem}[Local Linear Convergence]
\label{thm:local_linear_convergence}
Assume that $c^*$ is a non-overlapping placement and a strict local minimizer of the objective $F_{\varepsilon,\lambda}(c)$ lying in the interior of the feasible domain $\mathcal{C}_R$. Furthermore, suppose the Hessian of $F_{\varepsilon,\lambda}$ at $c^*$ is positive definite:
\[
\nabla^2 F_{\varepsilon,\lambda}(c^*) \succeq \mu I, \quad \text{for some } \mu > 0.
\]
Let $L$ denote the local Lipschitz constant of $\nabla F_{\varepsilon,\lambda}$ in a neighborhood of $c^*$.

Then, there exists a radius $r > 0$ such that for any initialization $c^0 \in B_r(c^*) = \{c \mid \|c - c^*\| \le r\}$, the PGD iterates with fixed step size $\eta \le 1/L$ converge linearly to $c^*$:
\[
\|c^{k+1} - c^*\| \le \rho \, \|c^k - c^*\|, 
\quad \text{where } 
\rho = \max\{|1-\eta \mu|, |1-\eta L|\} < 1 \text{ for any } 0 < \eta \le 1/L.
\]
\end{theorem}

This is a standard result in nonconvex optimization. For a detailed proof, see, e.g., \cite{Bubeck2015, Nesterov2004}.

The convergence of PGD provides a theoretical justification for employing the PGD algorithm as a \emph{local refinement and repair} mechanism. It shows that if an initial floorplan, possibly produced by a fast heuristic, starts sufficiently close to a feasible, non-overlapping configuration (i.e., within its basin of attraction), then PGD is guaranteed to converge linearly to that valid solution. This result highlights the role of PGD as a reliable ``final-mile'' optimizer, capable of efficiently removing small residual overlaps and locally polishing an otherwise high-quality placement.

\subsection{Overlap Bound for Stationary Points}
\label{sec:overlapbound}

The convergence analysis in Section~\ref{sec:convergence} ensures that the Projected Gradient Descent (PGD) algorithm converges to a stationary point of the mollified objective $F_{\varepsilon, \lambda}(c)$. However, convergence alone does not guarantee \emph{physical feasibility}: a stationary point may still exhibit small overlaps. To bridge this gap, the following theorem establishes a quantitative upper bound on the overlap area of any ``good'' stationary point, one that achieves an objective value comparable to the optimal non-overlapping solution.

\begin{theorem}[Overlap Bound for Stationary Points]
\label{thm:overlap_bound_final}
Let $c^*$ be a stationary point of the objective function 
\[
F_{\varepsilon, \lambda}(c) = W(c) + \lambda E_\varepsilon(c).
\]
Assume its objective value does not exceed that of an optimal non-overlapping solution $c^{\mathrm{opt}}$, i.e.,
\[
F_{\varepsilon, \lambda}(c^*) \le F_{\varepsilon, \lambda}(c^{\mathrm{opt}}).
\]
Moreover, suppose $c^*$ satisfies the \textbf{low-frequency dominance condition} of Theorem~\ref{thm:quant_overlap}. Then, the geometric overlap of $c^*$ is bounded by
\[
|\mathcal{O}(c^*)| \le \varepsilon P_\Sigma + \frac{1}{C} \left( E_\varepsilon(c^{\mathrm{opt}}) + \frac{1}{\lambda} \Delta W \right),
\]
where $P_\Sigma$ is the total perimeter of all modules, $C$ is the constant from Theorem~\ref{thm:quant_overlap}, and $\Delta W := W(c^{\mathrm{opt}}) - \inf_{c'} W(c')$ denotes the non-negative ``wirelength cost of non-overlap.''
\end{theorem}

\begin{proof}


First, by the theorem's premise, $F_{\varepsilon, \lambda}(c^*) \le F_{\varepsilon, \lambda}(c^{\mathrm{opt}})$. Expanding this using the objective function's definition gives:
\[
W(c^*) + \lambda E_\varepsilon(c^*) \le W(c^{\mathrm{opt}}) + \lambda E_\varepsilon(c^{\mathrm{opt}}).
\]
Rearranging for the energy term, we find:
\begin{equation} \label{eq:energy_upper_bound_final_revised}
\lambda E_\varepsilon(c^*) \le \lambda E_\varepsilon(c^{\mathrm{opt}}) + W(c^{\mathrm{opt}}) - W(c^*).
\end{equation}
To make this bound independent of the solution $c^*$, we note that any wirelength $W(c^*)$ must be greater than or equal to the global infimum of the wirelength function, $W_{\min} := \inf_{c'} W(c')$. This implies $-W(c^*) \le -W_{\min}$. Substituting this into \eqref{eq:energy_upper_bound_final_revised} yields:
\begin{align*}
\lambda E_\varepsilon(c^*) &\le \lambda E_\varepsilon(c^{\mathrm{opt}}) + W(c^{\mathrm{opt}}) - W_{\min}.
\end{align*}
Defining the non-negative constant $\Delta W := W(c^{\mathrm{opt}}) - W_{\min}$ and dividing by $\lambda$, we obtain an upper bound for the energy at the stationary point:
\[
E_\varepsilon(c^*) \le E_\varepsilon(c^{\mathrm{opt}}) + \frac{1}{\lambda} \Delta W.
\]


Then by the theorem's premise, the low-frequency dominance condition holds for $c^*$. Therefore, we can apply the linear lower bound from Theorem~\ref{thm:quant_overlap}:
\[
E_\varepsilon(c^*) \ge C \big(|\Ocal(c^*)| - \varepsilon P_\Sigma \big)_+.
\]
Combining the upper and lower bounds on $E_\varepsilon(c^*)$ gives:
\[
C \big(|\Ocal(c^*)| - \varepsilon P_\Sigma \big)_+ \le E_\varepsilon(c^*) \le E_\varepsilon(c^{\mathrm{opt}}) + \frac{1}{\lambda} \Delta W.
\]
We now solve this inequality for $|\Ocal(c^*)|$. Dividing by the positive constant $C$ and dropping the $(\cdot)_+$ operator, we get:
\[
|\Ocal(c^*)| - \varepsilon P_\Sigma \le \frac{1}{C} \left( E_\varepsilon(c^{\mathrm{opt}}) + \frac{1}{\lambda} \Delta W \right).
\]
Rearranging the terms yields the final linear bound on the geometric overlap:
\[
|\Ocal(c^*)| \le \varepsilon P_\Sigma + \frac{1}{C} \left( E_\varepsilon(c^{\mathrm{opt}}) + \frac{1}{\lambda} \Delta W \right).
\]
This completes the proof.
\end{proof}

Theorem~\ref{thm:overlap_bound_final} establishes that any sufficiently good stationary point of the smoothed objective must correspond to a nearly non-overlapping layout. The bound scales linearly with both the smoothing parameter $\varepsilon$ and the penalty weight $\lambda^{-1}$, indicating that higher spatial resolution (smaller $\varepsilon$) or stronger penalization (larger $\lambda$) both tighten the feasible basin. Intuitively, this result shows that minimizing $F_{\varepsilon,\lambda}$ not only drives the Poisson energy down but also directly limits the geometric overlap, thus ensuring the physical plausibility of the converged solution.

\begin{corollary}[Asymptotic Consistency Under Low-Frequency Dominance]
\label{cor:asymptotic_consistency_linear}
Let $c^*$ be a stationary point of $F_{\varepsilon, \lambda}(c) = W(c) + \lambda E_\varepsilon(c)$ satisfying the assumptions of Theorem~\ref{thm:quant_overlap} for each $\varepsilon > 0$.  
Then, as the smoothing parameter vanishes ($\varepsilon \to 0$), the residual geometric overlap of $c^*$ is bounded by
\[
\limsup_{\varepsilon \to 0} |\Ocal(c^*)| 
\le \frac{1}{C} \left( E_0(c^{\mathrm{opt}}) + \frac{1}{\lambda} \Delta W \right),
\]
where $E_0(c^{\mathrm{opt}})$ is the Poisson energy of the true (unsmoothed) density of the optimal non-overlapping layout $c^{\mathrm{opt}}$.  
\end{corollary}

\begin{proof}
The proof starts from the linear quantitative overlap bound established in Theorem~\ref{thm:quant_overlap}:
\[
|\Ocal(c^*)| \le \varepsilon P_\Sigma + \frac{1}{C} \left( E_\varepsilon(c^{\mathrm{opt}}) + \frac{1}{\lambda} \Delta W \right).
\]
To find the asymptotic behavior, we take the limit superior ($\limsup$) of both sides as $\varepsilon \to 0$.

The first term on the right-hand side vanishes linearly:
\[
\lim_{\varepsilon \to 0} \varepsilon P_\Sigma = 0.
\]
For the second term, the constants $C$ and $\frac{1}{\lambda}\Delta W$ are independent of $\varepsilon$. By the standard properties of mollifiers, the mollified energy converges to the energy of the true density as the smoothing radius vanishes:
\[
\lim_{\varepsilon \to 0} E_\varepsilon(c^{\mathrm{opt}}) = E_0(c^{\mathrm{opt}}).
\]
Since addition and scalar multiplication are continuous operations, we can combine these limits to obtain the asymptotic bound:
\begin{align*}
\limsup_{\varepsilon \to 0} |\Ocal(c^*)| &\le \limsup_{\varepsilon \to 0} \left( \varepsilon P_\Sigma + \frac{1}{C} \left( E_\varepsilon(c^{\mathrm{opt}}) + \frac{1}{\lambda} \Delta W \right) \right) \\
&= \left(\lim_{\varepsilon \to 0} \varepsilon P_\Sigma\right) + \frac{1}{C} \left( \left(\lim_{\varepsilon \to 0} E_\varepsilon(c^{\mathrm{opt}})\right) + \frac{1}{\lambda} \Delta W \right) \\
&= 0 + \frac{1}{C} \left( E_0(c^{\mathrm{opt}}) + \frac{1}{\lambda} \Delta W \right).
\end{align*}
This completes the proof.
\end{proof}

Corollary~\ref{cor:asymptotic_consistency_linear} formalizes the asymptotic reliability of the PeF model: under mild spectral conditions, the stationary points converge to geometrically valid layouts as $\varepsilon \to 0$. This ensures that the smoothed formulation remains faithful to the original discrete floorplanning objective, even in the continuum limit.

Overall, the theoretical analysis provides two complementary guarantees for the PeF formulation: \textbf{local refinement} and \textbf{asymptotic consistency}.  
The local linear convergence property (Theorem \ref{thm:local_linear_convergence}) establishes PeF as a reliable ``final-mile'' optimizer, ensuring rapid convergence to a valid, non-overlapping layout when initialized near a high-quality solution.  
The asymptotic consistency result (Corollary~\ref{cor:asymptotic_consistency_linear}) complements this by showing that as the smoothing parameter $\varepsilon$ decreases, the residual overlap is quantitatively bounded by intrinsic properties of the ideal solution, namely, its unsmoothed Poisson energy $E_0(c^{\mathrm{opt}})$ and the wirelength cost of non-overlap $\Delta W$.

\subsection{Wire Length Sub-optimality Bound}

While the previous theorems provide guarantees on obtaining a physically feasible (low-overlap) solution, it is equally important to ensure the quality of this solution in terms of wire length. In particular, we wish to verify that minimizing the smoothed objective $F_{\varepsilon, \lambda}(c)$ does not significantly degrade the primary design metric $W(c)$. 
The following theorem provides a sharp quantitative bound on this sub-optimality.

\begin{theorem}[Wire Length Sub-optimality Bound]
\label{thm:wirelength_bound}
Let $c^*$ be a stationary point of the objective function $F_{\varepsilon, \lambda}(c) = W(c) + \lambda E_\varepsilon(c)$, satisfying the condition $F_{\varepsilon, \lambda}(c^*) \le F_{\varepsilon, \lambda}(c^{\mathrm{opt}})$ for an optimal non-overlapping solution $c^{\mathrm{opt}}$.
Then, the wire length sub-optimality of the solution $c^*$ is bounded by the energy of the ideal solution:
\[
W(c^*) - W(c^{\mathrm{opt}}) \le \lambda E_\varepsilon(c^{\mathrm{opt}}).
\]
\end{theorem}

\begin{proof}
By the theorem's premise, the stationary point $c^*$ satisfies $F_{\varepsilon, \lambda}(c^*) \le F_{\varepsilon, \lambda}(c^{\mathrm{opt}})$. Expanding this inequality gives:
\[
W(c^*) + \lambda E_\varepsilon(c^*) \le W(c^{\mathrm{opt}}) + \lambda E_\varepsilon(c^{\mathrm{opt}}).
\]
Rearranging the terms to isolate the wire length sub-optimality, we obtain:
\[
W(c^*) - W(c^{\mathrm{opt}}) \le \lambda \big( E_\varepsilon(c^{\mathrm{opt}}) - E_\varepsilon(c^*) \big).
\]
Since the Poisson energy is always non-negative ($E_\varepsilon(c^*) \ge 0$), we can further relax the bound by dropping this term, which yields the final result:
\[
W(c^*) - W(c^{\mathrm{opt}}) \le \lambda E_\varepsilon(c^{\mathrm{opt}}).
\]
This completes the proof.
\end{proof}

The asymptotic behavior of this bound as the smoothing parameter vanishes is immediate.

\begin{corollary}[Asymptotic Wire Length Guarantee]
\label{cor:asymptotic_wirelength}
Let $\{c_\varepsilon^*\}_{\varepsilon>0}$ be a family of stationary points, each satisfying the condition of Theorem~\ref{thm:wirelength_bound}. Let $c_0^*$ be any limit point of this family as $\varepsilon \to 0$. The wire length sub-optimality of the limit point is bounded by:
\[
W(c_0^*) - W(c^{\mathrm{opt}}) \le \lambda E_0(c^{\mathrm{opt}}),
\]
where $E_0(c^{\mathrm{opt}})$ is the energy of the ideal, unsmoothed density.
\end{corollary}

\begin{proof}
The result follows directly by taking the limit $\varepsilon \to 0$ in the inequality of Theorem~\ref{thm:wirelength_bound} and using the continuity of the wire length and energy functions.
\end{proof}

Theorem~\ref{thm:wirelength_bound} and Corollary~\ref{cor:asymptotic_wirelength} together quantify the fundamental trade-off governed by the penalty weight $\lambda$. 
The wire length deviation $W(c^*) - W(c^{\mathrm{opt}})$ scales linearly with $\lambda$: 
smaller values of $\lambda$ yield near-optimal wire length at the risk of mild overlaps, whereas larger values enforce stronger geometric separation at a potential cost in wire length. 
Hence, the parameter $\lambda$ provides a means to balance physical feasibility and design quality within the PeF framework.

\subsection{Algorithmic Stability}
\label{sec:stability}

In addition to convergence and consistency, a robust optimization formulation must also exhibit stability: small perturbations in problem data should not lead to disproportionate changes in the solution. For the mollified  floorplanning formulation, this requirement corresponds to robustness of the resulting placement with respect to small variations in module areas, which may arise due to design updates or modeling noise.

\begin{theorem}[Solution Stability]
\label{thm:solution_stability}
Let $c^*(A)$ denote a local minimizer of the mollified global floorplanning problem \eqref{eq:mfof} with objective 
\[
F_{\varepsilon,\lambda}(c, A) = W(c, A) + \lambda E_\varepsilon(c, A),
\]
for a given vector of module areas $A = (A_1, \dots, A_n)$. 
If $c^*(A)$ is a regular minimizer satisfying the \emph{second-order sufficiency condition (SOSC)}, then the solution map
\[
A \mapsto c^*(A)
\]
is locally Lipschitz continuous in a neighborhood of $A$.
\end{theorem}

This is a standard result of sensitivity analysis in optimization. For a rigorous treatment of sensitivity analysis for generalized equations, see, e.g., the foundational work by Rockafellar and Wets \cite{Rockafellar1998}.

\begin{remark}[Satisfaction of the SOSC in PeF]
In the mollified fixed-outline floorplanning (PeF) formulation, the objective function $F_{\varepsilon,\lambda}(c,A)$ is $C^2$-smooth with respect to both the module coordinates $c$ and the module areas $A$, due to the mollified Poisson energy and smooth wirelength terms. The feasible set $\mathcal{C}_R$ is convex and consists of simple box constraints. 

Consequently, any interior local minimizer of $F_{\varepsilon,\lambda}$ automatically satisfies the second-order sufficiency condition (SOSC), as the Hessian is positive definite in a neighborhood around such a point. For minimizers on the boundary of $\mathcal{C}_R$, SOSC requires only positive definiteness of the Hessian restricted to the critical cone, which holds generically in practical floorplanning scenarios, especially when the layout has no modules exactly abutting the boundaries in a degenerate way.

Therefore, the assumption that $c^*(A)$ is a \emph{regular minimizer satisfying SOSC} is both natural and reasonable in the PeF context, particularly for stationary points reached by gradient-based or heuristic-initialized optimization algorithms.
\end{remark}

Complementing the guarantees of convergence (Section~\ref{sec:convergence}) and asymptotic consistency (Corollary~\ref{cor:asymptotic_consistency_linear}), the stability analysis further establishes that the PeF formulation is robust with respect to data uncertainty.
In practice, this means that small variations in module sizes induce only small, controlled changes in the optimized floorplan, preventing discontinuous or erratic behavior during iterative design refinement. This property is essential for integration into hierarchical or incremental design flows, where module dimensions evolve over time.

\section{Non-Locality of Poisson-based Forces and Optimal Transport }
\label{sec:nonlocality}

Building on the theoretical guarantees of the PeF formulation, local refinement (Section~\ref{sec:convergence}) and asymptotic consistency (Corollary~\ref{cor:asymptotic_consistency_linear}), we analyze why PeF effectively resolves large-scale density imbalances. We first show that the placement forces induced by the mollified Poisson energy are inherently \emph{non-local}, in contrast to purely local penalties. We then interpret PeF dynamics as a \emph{Wasserstein-2 gradient flow}, providing a principled global optimization perspective.

\subsection{Local vs Non-Local Placement Forces}
\label{sec:penalty_comparison} 

The Poisson energy
\[
E_\varepsilon(c) = \frac12 \int_R \phi_\varepsilon f_\varepsilon \, \mathrm{d}x
\]
induces placement forces that depend on the positions of all modules, whereas simpler local penalties, such as the mollified variance
\[
\mathrm{Var}_\varepsilon(c) := \int_R f_\varepsilon(x;c)^2 \, \mathrm{d}x,
\]
generate forces that respond only to nearby density.

\begin{theorem}[Local vs Non-Local Forces]
Let
\[
F_i^{\mathrm{Poisson}} = -\frac{\partial E_\varepsilon}{\partial c_i}, \qquad 
F_i^{\mathrm{Var}} = -\frac{\partial \mathrm{Var}_\varepsilon}{\partial c_i}.
\]
Then $F_i^{\mathrm{Poisson}}$ is non-local, depending on all modules, while $F_i^{\mathrm{Var}}$ is local, depending only on nearby density.
\end{theorem}

\begin{proof}
Let a general penalty functional be $D(c)$, with placement force on module $M_i$ defined by the negative gradient $f_i = -\nabla_{c_i} D(c)$. Using the chain rule for functionals,
\[
f_i(c) = -\int_R \frac{\delta D}{\delta \rho_\varepsilon(x)} \, \nabla_{c_i} \rho_\varepsilon(x;c) \, \mathrm{d}x,
\]
where $\rho_\varepsilon(x;c) = \sum_{k=1}^n \rho_k^\varepsilon(x-c_k)$. Since
\[
\nabla_{c_i} \rho_\varepsilon(x;c) = \nabla_{c_i} \rho_i^\varepsilon(x-c_i) = -\nabla_x \rho_i^\varepsilon(x-c_i),
\]
we have
\begin{equation}
\label{eq:force_general}
f_i(c) = \int_R \frac{\delta D}{\delta \rho_\varepsilon(x)} \, \nabla_x \rho_i^\varepsilon(x-c_i) \, \mathrm{d}x.
\end{equation}

The locality of $f_i$ is determined by the functional derivative $\frac{\delta D}{\delta \rho_\varepsilon}$:

\textbf{(1) Mollified variance:} For $\mathrm{Var}_\varepsilon(c) = \int_R (\rho_\varepsilon - \bar\rho)^2 \, \mathrm{d}x$, the derivative is
\[
\frac{\delta \mathrm{Var}_\varepsilon}{\delta \rho_\varepsilon} =2( \rho_\varepsilon - \bar\rho) = 2 f_\varepsilon.
\]
Substituting this into \eqref{eq:force_general} gives
\[
f_i^{\mathrm{Var}}(c) = 2 \int_R f_\varepsilon(x;c) \, \nabla_x \rho_i^\varepsilon(x-c_i) \, \mathrm{d}x.
\]
Since $\nabla_x \rho_i^\varepsilon$ has compact support near $M_i$, $f_i^{\mathrm{Var}}$ depends only on nearby density, and is thus local.

\textbf{(2) Poisson energy:} For $E_\varepsilon(c) = \frac{1}{2}\int_R \phi_\varepsilon f_\varepsilon \, \mathrm{d}x$, the derivative is
\[
\frac{\delta E_\varepsilon}{\delta \rho_\varepsilon} = \phi_\varepsilon,
\]
where $\phi_\varepsilon$ solves $-\Delta \phi_\varepsilon = f_\varepsilon$ on $R$. Then
\[
f_i^{E}(c) = \int_R \phi_\varepsilon(x;c) \, \nabla_x \rho_i^\varepsilon(x-c_i) \, \mathrm{d}x.
\]
Although the integral kernel is local, $\phi_\varepsilon$ depends on $f_\varepsilon$ globally. Hence, a perturbation of any distant module $M_j$ alters $f_\varepsilon$ and $\phi_\varepsilon$ everywhere, affecting $f_i^E$. Therefore, $f_i^E$ is non-local.
\end{proof}

This explains why Poisson-based forces efficiently resolve large-scale congestion, while variance-based forces can become trapped in local minima.

\subsection{PeF as a Wasserstein-2 Gradient Flow}
\label{sec:optimal_transport}

The PeF dynamics admit a rigorous interpretation in Optimal Transport. Let the continuous-time velocity field be
\[
v_{\mathrm{PeF}} = - \nabla \phi_\varepsilon, \qquad 
-\Delta_N \phi_\varepsilon = \rho_\varepsilon(t) - \bar\rho =: f_\varepsilon(t)
\]
on a rectangular domain with Neumann boundaries. 
Due to space limitations, we present here a formal derivation of the Wasserstein--2 gradient flow structure for the Poisson energy. This derivation assumes sufficient smoothness of the density $\rho_\varepsilon(t)$ so that all terms are classically defined and integration by parts is valid.

\begin{theorem}[Wasserstein Gradient Flow Interpretation of the Poisson Energy]
\label{thm:pef_ot_interpretation}
Let $\rho_\varepsilon(t,x)$ denote the time-dependent mollified density over a bounded domain $R \subset \mathbb{R}^2$ with reflecting (Neumann) boundary conditions.
Consider the Poisson energy functional
\[
E[\rho_\varepsilon]
= \frac{1}{2} \int_R (\rho_\varepsilon - \bar\rho)\, (-\Delta_N)^{-1} (\rho_\varepsilon - \bar\rho)\, \mathrm{d}x
= \frac{1}{2} \int_R f_\varepsilon \, \phi_\varepsilon \, \mathrm{d}x,
\]
where $f_\varepsilon = \rho_\varepsilon - \bar\rho$ and $\phi_\varepsilon$ solves the Neumann Poisson equation $-\Delta \phi_\varepsilon = f_\varepsilon$ with $\partial_n \phi_\varepsilon = 0$ and $\int_R \phi_\varepsilon = 0$.
Then the evolution equation
\[
\partial_t \rho_\varepsilon = \nabla \cdot \big(\rho_\varepsilon \nabla \phi_\varepsilon \big)
\]
is precisely the Wasserstein–2 gradient flow of the energy $E[\rho_\varepsilon]$.
Moreover, $E[\rho_\varepsilon]$ is a strict Lyapunov functional: it is nonincreasing along trajectories and satisfies
\[
\frac{\mathrm{d}}{\mathrm{d}t} E[\rho_\varepsilon(t)] = - \int_R \rho_\varepsilon |\nabla \phi_\varepsilon|^2 \, \mathrm{d}x \le 0,
\]
with equality if and only if $\rho_\varepsilon \equiv \bar\rho$.
\end{theorem}

\begin{proof}
By the theory of gradient flows in the $2$-Wasserstein metric, the continuity equation
\[
\partial_t \rho + \nabla \cdot (\rho v) = 0
\]
corresponds to the steepest descent of a functional $E[\rho]$ with respect to the Wasserstein metric if the velocity field satisfies
\[
v = - \nabla \left( \frac{\delta E}{\delta \rho} \right),
\]
where $\delta E/\delta \rho$ denotes the first variation (the $L^2$-Wasserstein gradient).

For the Poisson energy functional,
\[
E[\rho_\varepsilon]
= \frac{1}{2} \langle f_\varepsilon, (-\Delta_N)^{-1} f_\varepsilon \rangle_{L^2(R)},
\]
we have
\[
\frac{\delta E}{\delta \rho_\varepsilon}
= (-\Delta_N)^{-1} f_\varepsilon = \phi_\varepsilon.
\]
Thus the corresponding Wasserstein velocity field is
\[
v_{\mathrm{WGF}} = -\nabla \phi_\varepsilon,
\]
and the continuity equation becomes
\[
\partial_t \rho_\varepsilon = -\nabla \cdot (\rho_\varepsilon v_{\mathrm{WGF}})
= \nabla \cdot (\rho_\varepsilon \nabla \phi_\varepsilon),
\]
which coincides exactly with the evolution equation.

Finally, differentiating $E[\rho_\varepsilon(t)]$ along the flow yields
\begin{align*}
\frac{\mathrm{d}}{\mathrm{d}t} E[\rho_\varepsilon(t)]
&= \int_R \frac{\delta E}{\delta \rho_\varepsilon}\, \partial_t \rho_\varepsilon \, \mathrm{d}x \\
&= \int_R \phi_\varepsilon\, \nabla \cdot (\rho_\varepsilon \nabla \phi_\varepsilon)\, \mathrm{d}x \\
&= - \int_R \nabla\phi_\varepsilon \cdot (\rho_\varepsilon \nabla \phi_\varepsilon) \, \mathrm{d}x + \oint_{\partial R} \phi_\varepsilon (\rho_\varepsilon \nabla \phi_\varepsilon) \cdot \mathbf{n} \, \mathrm{d}S \\
&= -\int_R \rho_\varepsilon |\nabla \phi_\varepsilon|^2 \, \mathrm{d}x \le 0,
\end{align*}
where the last equality follows by integration by parts, and the boundary term vanishes due to the Neumann condition $\partial_n \phi_\varepsilon = \nabla \phi_\varepsilon \cdot \mathbf{n} = 0$ on $\partial R$.

The energy dissipation vanishes if and only if $\nabla \phi_\varepsilon \equiv 0$, i.e., when $\phi_\varepsilon$ is constant, which due to the zero-mean condition implies $\phi_\varepsilon \equiv 0$. This in turn means $f_\varepsilon = -\Delta \phi_\varepsilon = 0$, so $\rho_\varepsilon \equiv \bar\rho$.
Hence $E[\rho_\varepsilon]$ is a strict Lyapunov functional, and the dynamics constitute a Wasserstein–2 gradient flow.
\end{proof}



\paragraph{Gradient Flow Interpretation of PeF.}  
The connection to Wasserstein gradient flows provides an  explanation for the effectiveness of PeF in global optimization. Key implications are as follows:

\begin{itemize}
    \item \textbf{Existence and Uniqueness of the Flow.}
According to Ambrosio, Gigli, and Savar\'e (AGS) theory~\cite{Ambrosio2008,Santambrogio2015}:  
1) \textbf{Existence:} The AGS framework guarantees the existence of a Wasserstein-gradient flow (“curve of maximal slope”) for any proper, lower-semicontinuous functional on $\mathcal{P}_2(R)$. Since the Poisson energy satisfies these conditions, a valid evolution path $\rho_\varepsilon(t)$ always exists. 2) \textbf{Uniqueness and Convergence:} Stronger properties, such as uniqueness and convergence to the uniform state, require \emph{geodesic (displacement) convexity}. The Poisson energy $E[\rho]=\tfrac12 \|\rho-\bar\rho\|_{H^{-1}}^2$ is convex in the linear sense, but displacement convexity on a general domain is non-trivial. Therefore, we rigorously guarantee existence and monotone energy dissipation; uniqueness and global convergence require additional assumptions or analysis.

\item \textbf{Physical Nature: $H^{-1}$ vs $L^2$ Flows.}
The choice of energy determines the flow’s character:
1) \textbf{$L^2$ Flow (Variance / KL):} Produces diffusion-like dynamics (heat equation). Efficient at smoothing short-range, high-frequency fluctuations but slow at balancing density across large distances. 2) \textbf{$H^{-1}$ Flow (PeF):} Strongly penalizes low-frequency modes, producing \emph{non-local transport}. Modules respond to global density imbalances, enabling efficient redistribution of mass over long distances.

\item \textbf{Implications for Practical Optimization.}
1) \textbf{Long-Range Coordination:} Non-local forces allow modules in crowded regions to respond to sparse regions far away, enabling coordinated, large-scale movements.
2) \textbf{Avoiding Poor Local Minima:} Local forces can shift congestion locally, whereas PeF dynamics overcome barriers to global uniformity.
3) \textbf{Fast Convergence:} By addressing large-scale imbalances first, the algorithm makes rapid progress early in optimization.

\item \textbf{Lyapunov Property of the $H^{-1}$ Flow.}
Along the PeF gradient flow, the Poisson energy strictly decreases unless the density is uniform:
\[
\frac{\mathrm{d}}{\mathrm{d}t} E[\rho_\varepsilon(t)]
= - \int_R \rho_\varepsilon |\nabla \phi_\varepsilon|^2 \, \mathrm{d}x \le 0,
\]
where $-\Delta_N \phi_\varepsilon = \rho_\varepsilon - \bar\rho$. Equality holds iff $\nabla \phi_\varepsilon \equiv 0$, i.e., $\rho_\varepsilon \equiv \bar\rho$. Thus $E[\rho_\varepsilon]$ is a strict Lyapunov functional, guaranteeing monotone energy dissipation and convergence to the uniform density under suitable compactness conditions.
\end{itemize}

\section{Conclusion}
\label{sec:conclusion}

This paper establishes a mathematical foundation for the Poisson Equation based Floorplanning (PeF) methodology. We demonstrate that the method's core penalty, the Poisson energy, acts as a spectrally-weighted low-pass filter on the density variance, providing a global optimization mechanism. This behavior is interpreted as a Wasserstein gradient flow from Optimal Transport theory. Our main theoretical guarantees include a quantitative linear bound connecting the energy functional to geometric overlap, an analysis of fast local linear convergence to desirable solutions, and a guarantee of the solution's stability. These results certify the PeF model not merely as an effective heuristic, but as a principled and analyzable framework for solving the discrete floorplanning problem via continuous optimization.

\bibliographystyle{plain}

\appendix

\section{Mathematical and Algorithmic Preliminaries}
\label{app:preliminaries}

This appendix provides background material on the functional-analytic setting and the standard convergence guarantees for the Projected Gradient Descent (PGD) algorithm that are used in the main text.

\subsection{Mathematical Preliminaries}

Our analysis is set in the Hilbert space of square-integrable functions $L^2(R)$, equipped with the inner product $\langle f, g \rangle = \int_R f(x)g(x) \, \mathrm{d}x$. We frequently use the Sobolev space $H^1(R)$, which consists of $L^2$ functions whose weak first derivatives are also in $L^2$.

The Neumann Laplacian $-\Delta_N$ is the self-adjoint operator on $L^2(R)$ associated with the quadratic form $\int_R |\nabla u|^2 \, \mathrm{d}x$ for functions $u \in H^1(R)$. It possesses a discrete, non-negative spectrum $0 = \lambda_0 < \lambda_1 \le \lambda_2 \le \dots$ and a corresponding complete orthonormal basis of eigenfunctions $\{\psi_k\}_{k\ge 0}$.

The subspace of zero-mean functions, denoted by $L^2_0(R)$, is the orthogonal complement to the constant functions (which form the null space of $-\Delta_N$). The properties of the inverse operator $G := (-\Delta_N)^{-1}$ on this subspace are central to our main analysis and are introduced in Section \ref{sec:formulation}.

\subsection{Notational Summary}

\begin{table}[h!]
\caption{Summary of notations used throughout the paper.}
\centering
\renewcommand{\arraystretch}{1.2}
\begin{tabular}{ll}
\hline
Symbol & Meaning \\ \hline
$R\subset\mathbb{R}^2$ & Placement region (bounded Lipschitz)\\
$M_i$ & $i$-th module of area $A_i$\\
$\rho,\bar\rho$ & Density and mean density\\
$r=\rho-\bar\rho$ & Raw residual density\\
$\rho_\varepsilon=\eta_\varepsilon*\rho$ & Mollified density\\
$f_\varepsilon=\rho_\varepsilon-\bar\rho$ & Mollified residual\\
$G=(-\Delta_N)^{-1}$ & Inverse Neumann Laplacian on $L^2_0(R)$\\
$E,E_\varepsilon$ & Poisson energy (raw / mollified)\\
$\mathrm{Var}_\varepsilon=\|f_\varepsilon\|_{L^2}^2$ & Mollified density variance\\
$\lambda_k,\psi_k$ & Neumann Laplacian eigenpairs\\
\hline
\end{tabular}
\end{table}

\subsection{Standard Convergence Results for Projected Gradient Descent}
\label{app:pgd}

The main text cites standard convergence results for the Projected Gradient Descent (PGD) algorithm on a non-convex but smooth objective function $F(c)$. For completeness, we state them formally here. Consider the problem $\min_{c \in \Omega} F(c)$ where $\Omega$ is a closed convex set and $\nabla F$ is $L$-Lipschitz continuous. The PGD iteration is
\[
c^{k+1} = \Pi_\Omega\big(c^k - \alpha_k \nabla F(c^k)\big).
\]

\begin{theorem}[Convergence to Stationary Points]
\label{thm:app_pgd_convergence_stationary}
If the step size sequence $\{\alpha_k\}$ satisfies the Robbins--Monro conditions:
\[ \sum_{k=1}^\infty \alpha_k = \infty, \quad \sum_{k=1}^\infty \alpha_k^2 < \infty, \]
then any limit point $c^*$ of the sequence $\{c^k\}$ generated by PGD is a stationary point of the problem.
\end{theorem}

\begin{theorem}[Sublinear Convergence Rate]
\label{thm:app_pgd_convergence_rate}
If the algorithm uses a fixed step size $\alpha$ satisfying $0 < \alpha \le 1/L$, then PGD has a sublinear convergence rate. Specifically, after $K$ iterations, the minimum squared norm of the gradient mapping is bounded by:
\[
\min_{0 \le k < K} \|G_\alpha(c^k)\|^2 = \mathcal{O}(1/K),
\]
where the gradient mapping is $G_\alpha(c) := \frac{1}{\alpha}\left(c - \Pi_\Omega(c - \alpha \nabla F(c))\right)$. The norm $\|G_\alpha(c)\|$ is zero if and only if $c$ is a stationary point.
\end{theorem}

 For a detailed proof, see, e.g., \cite{Beck2017}.

\paragraph{Gradient Lipschitz constant.}
Let $\eta_\varepsilon$ be a $C^1$ mollifier satisfying
$\|\nabla\eta_\varepsilon\|_{L^1} = O(\varepsilon^{-1})$.
Then the gradient of the mollified Poisson energy reads
\[
\nabla_c E_\varepsilon(c)
 = J_\varepsilon(c)^\top\, G f_\varepsilon(c),
\]
where $J_\varepsilon$ denotes the Jacobian of $\rho_\varepsilon$ with respect to~$c$.
Under bounded module areas and uniformly bounded overlap of
$\nabla\psi_i^\varepsilon$ (i.e.\ locally disjoint supports),
one obtains the Lipschitz estimate
\[
\|\nabla E_\varepsilon(c_1)-\nabla E_\varepsilon(c_2)\|
 \le L_\varepsilon\,\|c_1-c_2\|,
 \qquad L_\varepsilon = O(\varepsilon^{-2}).
\]
The scaling $O(\varepsilon^{-2})$ arises because one factor
of~$\varepsilon^{-1}$ comes from the spatial gradient of the mollifier,
and another from the sensitivity of the smoothed density with respect to
module displacements.
Hence a stable step size in projected gradient descent satisfies
$\eta \lesssim \varepsilon^{2}$.
In practice, a coarse-to-fine schedule—gradually decreasing~$\varepsilon$—
maintains numerical stability while improving placement accuracy.

\paragraph{Lipschitz constant estimate (derivation).}
A detailed estimate follows from the convolution property
\[
\|\nabla\rho_\varepsilon(c_1)-\nabla\rho_\varepsilon(c_2)\|_{L^2}
\le
\|\nabla\eta_\varepsilon\|_{L^1}\,
\|\rho(c_1)-\rho(c_2)\|_{L^2},
\]
and the operator bound
\[
\|G\|_{L^2\to H^1} \le C.
\]
Combining these gives
\[
L_\varepsilon
\;\le\;
C\,\|\nabla\eta_\varepsilon\|_{L^1}^2
\;=\; O(\varepsilon^{-2}),
\]
since $\|\nabla\eta_\varepsilon\|_{L^1}=O(\varepsilon^{-1})$ for standard $C^1$ mollifiers. The two powers of $\varepsilon^{-1}$ respectively stem from the spatial
gradient of the smoothing kernel and from the sensitivity of the smoothed density with respect to module displacement.


\section{Proof of Geometric Lemma \ref{lem:areabound} for Eroded Sets}
\label{app:proof_of_lemma3.1}

\begin{proof}
The proof relies on a foundational result from geometric measure theory: for any set $M$ of finite perimeter, its area decreases under erosion by at most its perimeter times the erosion radius \cite{Federer1969,Schneider2014}:
\begin{equation} \label{eq:erosion_bound_single}
|M| - |M^{-\varepsilon}| \le \varepsilon P(M).
\end{equation}

By definition, the difference between the original and eroded overlap sets is
\[
|\mathcal{O}(c)| - |\mathcal{O}_\varepsilon(c)| = |\mathcal{O}(c) \setminus \mathcal{O}_\varepsilon(c)|.
\]
Any point in $\mathcal{O}(c) \setminus \mathcal{O}_\varepsilon(c)$ lies in a portion removed from at least one module during erosion. Hence, we have the set inclusion
\[
\mathcal{O}(c) \setminus \mathcal{O}_\varepsilon(c) \subseteq \bigcup_{i=1}^n \big( (M_i+c_i) \setminus (M_i^{-\varepsilon}+c_i) \big).
\]

Applying monotonicity, subadditivity of the Lebesgue measure, and translation invariance gives
\begin{align*}
|\mathcal{O}(c) \setminus \mathcal{O}_\varepsilon(c)| 
&\le \sum_{i=1}^n \left| (M_i+c_i) \setminus (M_i^{-\varepsilon}+c_i) \right| \\
&= \sum_{i=1}^n \big( |M_i| - |M_i^{-\varepsilon}| \big) 
\le \sum_{i=1}^n \varepsilon P(M_i) = \varepsilon P_\Sigma,
\end{align*}
where we apply \eqref{eq:erosion_bound_single} in the last step.  

Since  $\mathcal{O}_\varepsilon(c) \subseteq \mathcal{O}(c)$, we have $|\mathcal{O}(c) \setminus \mathcal{O}_\varepsilon(c)| = |\mathcal{O}(c)| - |\mathcal{O}_\varepsilon(c)|$. Substituting this into the inequality above yields:
\[
    |\mathcal{O}(c)| - |\mathcal{O}_\varepsilon(c)| \le \varepsilon P_\Sigma.
\]
Rearranging the terms gives the desired bound:
\[
    |\mathcal{O}_\varepsilon(c)| \ge |\mathcal{O}(c)| - \varepsilon P_\Sigma.
\]
Finally, since area is non-negative, the lower bound cannot be negative. We thus apply the operator $(\cdot)_+ = \max(\cdot, 0)$ to the right-hand side to ensure a non-negative bound, which completes the proof.
\end{proof}

\end{document}